\documentclass[10pt,conference]{IEEEtran}
\usepackage{maxim}
\usepackage{graphicx}
\usepackage{cite}

\def\cA{{\cal A}}

\def\cF{{\cal F}}
\def\cG{{\cal G}}
\def\cM{{\cal M}}
\def\cN{{\cal N}}
\def\cP{{\cal P}}
\def\cX{{\cal X}}
\def\cY{{\cal Y}}
\def\cZ{{\cal Z}}

\def\DD{{\mathbb D}}

\def\emp{{\mathsf P}}

\def\wh#1{\widehat{#1}}

\def\KS{{\rm KS}}

\begin{document}
\title{Achievability Results for Statistical Learning\\
Under Communication Constraints}
\author{
\authorblockN{Maxim Raginsky}
\authorblockA{Department of Electrical and Computer Engineering\\
Duke University, Durham, NC 27708, USA\\
Email: m.raginsky@duke.edu}}

\maketitle

\begin{abstract} The problem of statistical learning is to construct an accurate predictor of a random variable as a function of a correlated random variable on the basis of an i.i.d.\ training sample from their joint distribution. Allowable predictors are constrained to lie in some specified class, and the goal is to approach asymptotically the performance of the best predictor in the class. We consider two settings in which the learning agent only has access to rate-limited descriptions of the training data, and present information-theoretic bounds on the  predictor performance achievable in the presence of these communication constraints. Our proofs do not assume any separation structure between compression and learning and rely on a new class of operational criteria specifically tailored to joint design of encoders and learning algorithms in rate-constrained settings.
\end{abstract}

\thispagestyle{empty}

\section{Introduction}
\label{sec:intro}

Let $X \in \cX$ and $Y \in \cY$ be jointly distributed random variables. The problem of statistical learning is to design an accurate predictor of the {\em output variable} $Y$ from the {\em input variable} $X$ on the basis of a number of independent {\em training samples} drawn from their joint distribution, with very little or no prior knowledge of that distribution. The present paper focuses on the achievable performance of learning schemes when the learning agent only has access to a finite-rate description of the training samples.

This problem of {\em learning under communication constraints} arises in a variety of contexts, such as distributed estimation using a sensor network, adaptive control, or repeated games. In these and other scenarios, it is often the case that the agents who gather the training data are geographically separated from the agents who use these data to make inferences and decisions, and communication between these two types of agents is possible only over rate-limited channels. Hence, there is a trade-off between the communication rate and the quality of the inference, and it is of interest to characterize this trade-off mathematically.

This paper follows on our earlier work \cite{Rag07a} and presents improved bounds on the achievable performance of statistical learning schemes operating under two kinds of communication constraints: (a) the entire training sequence is delivered to the learning agent over a rate-limited noiseless digital channel, and (b) the input part of the training sequence is available to the learning agent with arbitrary precision, while the output part is delivered, as before, over a rate-limited channel. Whereas \cite{Rag07a} has looked at schemes where the finite-rate description of the training data was obtained through vector quantization, effectively imposing a separation structure between compression and learning, here we remove this restriction.  

We show that, under certain regularity conditions, there is no penalty for compression of the training sequence in the setting (a). This is due to the fact that the encoder can reliably estimate the underlying distribution (in the metric specifically tailored for the learning problem at hand) and then communicate the finite-rate description to the learning agent, who can then find the optimum predictor for the estimated distribution. The setting (b), however, is radically different: because the encoder has no access to the input part of the training sample, it cannot estimate the underlying distribution. Instead, the encoder constructs a finite-rate description of the output part using a specific kind of a vector quantizer, namely one designed to minimize the expected distance between the underlying distribution (whatever it may happen to be) and the empirical distribution of the input/quantized output pairs. Our achievability result for the setting (b) uses a learning-theoretic generalization of recent work by Kramer and Savari \cite{KraSav07} on rate-constrained communication of probability distributions.

The problem of learning a pattern classifier under rate constraints was also treated in a recent paper by Westover and O'Sullivan \cite{WesSul08}. They assumed that the underlying probability distribution is known, and the rate constraint arises from the limitations on the memory of the learning agent; then the problem is to design the best possible classifier (without any constraints on its structure). The motivation for the work in \cite{WesSul08} comes from biologically inspired models of learning. The approach of the present paper is complementary to that of \cite{WesSul08}. We consider a more general, decision-theoretic formulation of learning that includes regression as well as classification, but allow only vague prior knowledge of the underlying distribution and assume that the class of available predictors is constrained. Thus, while \cite{WesSul08} presents information-theoretic bounds on the performance of {\em any} classifier (including ones that are fully cognizant of the generative model for the data), here we are concerned with the performance of constrained learning schemes that must perform well in the presence of uncertainty about the underlying distribution.

The novel element of our approach is that both the operational criteria used to design the encoders and the learning algorithm, and the regularity conditions that must hold for rate-constrained learning to be possible, involve a tight coupling between the available prior knowledge about the underlying distribution and the set of predictors available to the learning agent. Planned future work includes obtaining converse theorems (lower bounds) and applying our formalism to specific classes of predictors used in statistical learning theory.

\section{Preliminaries and problem formulation}
\label{sec:prelims}

A very general decision-theoretic formulation of the learning problem, due to Haussler \cite{Hau92}, goes as follows. We have a family $\cP$ of probability distributions on $\cZ \deq \cX \times \cY$ and a class $\cF$ of measurable functions $f : \cZ \to \R$. For any $P \in \cP$, define
$$
L(f,P) \deq \E_P [f(Z)] \equiv \int_\cZ f(z) dP(z), \qquad f \in \cF
$$
and
$$
L^*(\cF,P) \deq \inf_{f \in \cF} L(f,P),
$$
where we assume that the infimum is achieved by some $f^* \in \cF$. The family $\cP$ represents prior knowledge about the joint distribution of $X$ and $Y$; each function $f \in \cF$ corresponds to the loss incurred by a particular predictor of $Y$ based on $X$. This framework covers, for instance, the following standard scenarios:
\begin{itemize}
\item {\em classification} --- $\cX \subseteq \R^d$, $\cY = \{1,\ldots,M\}$, and $\cF$ consists of functions of the form
$$
f(x,y) =  I_{\{g(x) \neq y\}}, \qquad g \in \cG
$$
where $I_{\{\cdot\}}$ is the indicator function, and $\cG$ is a given family of {\em classifiers}, i.e., measurable functions $g : \cX \to \{1,\ldots,M\}$. Any $f^* \in \cF$ that achieves $L^*(\cF,P)$ corresponds to some $g^* \in \cG$ that has the smallest classification error: $P(g^*(X) \neq Y) = \inf_{g \in \cG} P(g(X) \neq Y)$.

\item {\em regression} --- $\cX \subseteq \R^d$, $\cY \subseteq \R$, and $\cF$ consists of functions of the form
$$
f(x,y) = (g(x) - y)^2, \qquad g \in \cG
$$
where $\cG$ is a given family of {\em estimators}, i.e., measurable functions $g : \cX \to \R$. Any $f^* \in \cF$ that achieves $L^*(\cF,P)$ corresponds to some $g^* \in \cG$ that has the smallest mean squared error: $\E_P[(g^*(X)-Y)^2] = \inf_{g \in \cG} \E_P[(g(X) - Y)^2]$.
\end{itemize}
These are instances of {\em supervised learning} problems. {\em Unsupervised} settings, where $\cY = \varnothing$ (such as density estimation or clustering),  can also be accommodated by Haussler's framework. In this paper we focus only on the supervised case; thus, we will assume that $|\cY| \ge 2$. Then the learning problem is to construct, for each $n \in \N$, an approximation to $f^*$ on the basis of a {\em training sequence} $Z^n = \{Z_i\}^n_{i=1}$, where $Z_i = (X_i,Y_i)$ are i.i.d.\ according to some unknown $P \in \cP$.

Formally, a {\em learning scheme} (or {\em learner}, for short) is a sequence $\{ \wh{f}_n \}^\infty_{n=1}$ of maps $\wh{f}_n : \cZ^n \times \cZ \to \R$, such that $\wh{f}_n(z^n,\cdot) \in \cF$ for all $z^n \in \cZ^n$. Let $Z = (X,Y) \sim P$ be independent of the training sequence $Z^n$. The main quantity of interest is the {\em generalization error}
$$
L(\wh{f}_n,P) = \E\Big[ \wh{f}_n(Z^n,Z) \Big| Z^n \Big] \equiv \int_\cZ \wh{f}_n(Z^n,z) dP(z),
$$
which is a random variable that depends on the training sequence $Z^n$. Under suitable regularity conditions on $\cP$ and $\cF$, one can show that there exist learning schemes that are {\em probably approximately correct} (PAC), i.e., for every $\epsilon > 0$ and $P \in \cP$,
\begin{equation}
\lim_{n \to \infty} P \left( Z^n : L(\wh{f}_n,P) > L^*(\cF,P) + \epsilon \right) = 0
\label{eq:PAC}
\end{equation}
(see, e.g.,\ Vidyasagar \cite{Vid03}). A more modest goal is to ensure that the {\em excess loss} $L(\wh{f}_n,P) - L^*(\cF,P)$ is small, either in probability or in expectation.

We are interested in the achievable excess loss in situations where there is a rate-constrained channel between the source of the training data and the learning agent. Specifically, we shall consider the following two scenarios, depicted in Figs.~\ref{fig:type_I} and \ref{fig:type_II}, respectively.

\begin{figure}[htb]
\centerline{\includegraphics[width=0.8\columnwidth]{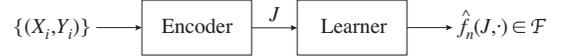}}
\caption{\label{fig:type_I} Type I set-up: the encoder has full observation of the training samples.}
\end{figure}

In the first set-up, shown in Fig.~\ref{fig:type_I}, the learner observes the training data through a noiseless digital channel that can transmit a fixed finite number of bits per training pair $Z = (X,Y)$. A scheme for learning operating at rate $R$ is specified by a sequence $\{(e_n,\wh{f}_n)\}^\infty_{n=1}$, where $e_n : \cZ^n \to \{1,2,\ldots,M_n\}$ is the {\em encoder} and $\wh{f}_n : \{1,2,\ldots,M_n\} \to \cF$ is the {\em learner}, such that $\limsup_{n \to \infty} n^{-1} \log M_n \le R$. For each $n$, the output of the learner is a function $\wh{f}_n(J,\cdot) \in \cF$, where $J = e_n(Z^n)$ is the finite-rate description of $Z^n$ provided by the encoder. We shall refer to this as Type I set-up.

\begin{figure}[htb]
\centerline{\includegraphics[width=0.8\columnwidth]{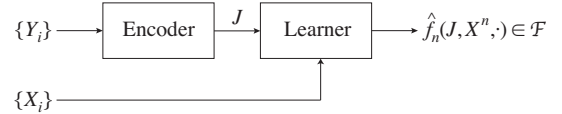}}
\caption{\label{fig:type_II} Type II: the encoder sees only the output  part of the training sequence.}
\end{figure}

In the second set-up, shown in Fig.~\ref{fig:type_II}, the learner has perfect observation of the input ($\cX$-valued) part of the training sequence, while the output ($\cY$-valued part) is delivered over a rate-limited noiseless digital channel. A scheme for learning operating at rate $R$ is  a sequence $\{(e_n,\wh{f}_n)\}^\infty_{n=1}$, where $e_n : \cY^n \to \{1,2,\ldots,M_n\}$ is the encoder and $\wh{f}_n : \cX^n \times \{1,2,\ldots,M_n\} \to \cF$ is the learner, such that $\limsup_{n \to \infty} n^{-1} \log M_n \le R$. For each $n$, the output of the learner is a function $\wh{f}_n(J,X^n,\cdot) \in \cF$, where $J = e_n(Y^n)$ is the finite-rate description of $Y^n$ provided by the encoder.

\sloppypar We shall often abuse notation and let $\wh{f}_n$ denote also the function in $\cF$ returned by the learner. The main object of interest is the generalization error
$$
L(e_n,\wh{f}_n,P) \deq \E \Big[ \wh{f}_n(W_n,Z) \Big| Z^n \Big], \qquad P \in \cP
$$
\sloppypar\noindent where $Z = (X,Y) \sim P$ is assumed independent of $\{Z_i\}^n_{i=1}$, and $W_n$ is equal to $J = e_n(Z^n)$ in a Type I set-up and to $(J,X^n)$ in a Type II set-up, where  $J = e_n(Y^n)$. We are interested in the achievable values of the asymptotic expected excess loss. We say that a pair $(R,\Delta)$ is {\em achievable for $(\cF,\cP)$} if there exists a scheme $\{(e_n,\wh{f}_n)\}^\infty_{n=1}$ operating at rate $R$, such that
$$
\limsup_{n \to \infty} \E L(e_n,\wh{f}_n,P) \le L^*(\cF,P) + \Delta, \qquad \forall P \in \cP.
$$

\section{Achievability theorems}
\label{sec:achieve}

In this section, we prove two theorems about achievable pairs $(R,\Delta)$ in Type I and Type II settings. The key idea in both cases is that the encoder needs to provide enough information at rate $R$ for the learner to estimate the expected value of each $f \in \cF$ to within $\Delta$.

\subsection{Notation, preliminaries and assumptions}
\label{ssec:notation}

We assume that the space $\cZ$ is equipped with an appropriate $\sigma$-algebra $\cA$. Typical cases of interest in learning theory are $\cX \subset \R^d$ and $\cY$ finite (classification) or $\cX \subseteq \R^d$ and $\cY \subseteq \R$ (regression), with the usual Borel $\sigma$-algebras. The space of all probability measures on $(\cZ,\cA)$ will be denoted by $\cM(\cZ)$. $\cF$ is a class of measurable functions from $(\cZ,\cA)$ into $[0,B]$ for some $0 < B < +\infty$; to avoid various measurability issues, we also assume throughout that $\cF$ is countable.  We shall identify signed measures $\mu$ on $(\cZ,\cA)$ with real-valued linear functionals $f \mapsto \mu(f)$ on $\cF$, where $\mu(f) \deq \int_\cZ f d\mu$. Thus, to each $\mu$ we can associate the $\ell^\infty(\cF)$-norm
$$
\| \mu \|_\cF \deq \sup_{f \in \cF} |\mu(f)|.
$$
For an $n$-tuple $z^n \in \cZ^n$, $\emp_{z^n}$ will denote the corresponding empirical measure: $\emp_{z^n} = n^{-1}\sum^n_{i=1} \delta_{z_i}$, where $\delta_z$ is the Dirac measure (point mass) concentrated at $z \in \cZ$. We assume that $\cF$ is a {\em Glivenko--Cantelli (GC) class} \cite{WaaWel96}, i.e.,
\begin{equation}
\lim_{n \to \infty} \| \emp_{Z^n} - P \|_\cF = 0, \qquad \mbox{a.s.}
\label{eq:GC}
\end{equation}
for every $P \in \cM(\cZ)$. In other words, the class $\cF$ is such that, for each $P \in \cM(\cZ)$, the sample averages $\emp_{Z^n}(f)$ converge to the theoretical averages $P(f)$ uniformly over $\cF$. This is a standard assumption in statistical learning theory \cite{WaaWel96,Vid03}.

\subsection{Type I schemes}
\label{ssec:type_I}

We now show that, in a Type I set-up, there is no penalty for compression of the training sequence, provided the family $\cP$  is not too ``rich." Our notion of richness will pertain to the geometry of $\cP$ w.r.t.\ the $\| \cdot \|_\cF$ norm. Given some $\epsilon > 0$, we say that a finite set $\{P_1,\ldots,P_M\} \subset \cP$ is an {\em $\epsilon$-net} for $\cP$ if
$$
\sup_{P \in \cP} \min_{1 \le m \le M} \| P - P_m \|_\cF \le \epsilon.
$$
We define the {\em covering number} $N_\cF(\epsilon,\cP)$ as the cardinality of the minimal $\epsilon$-net of $\cP$, and the {\em Kolmogorov $\epsilon$-entropy} of $\cP$ as $H_\cF(\epsilon,\cP) \deq \log N_\cF(\epsilon,\cP)$ \cite{KolTih61}.

\begin{theorem} Suppose that there exists a monotone decreasing sequence $\{\epsilon_n\}^\infty_{n=1}$ of nonnegative reals, such that
\begin{equation}
H_\cF(\epsilon_n,\cP) = o(n).
\label{eq:entropy_condition}
\end{equation}
Then the pair $(0,0)$ is achievable for $(\cF,\cP)$.
\label{thm:type_I}
\end{theorem}

\begin{proof} For each $n$, let $\cN_n = \{P_1,P_2,\ldots,P_{M_n}\}$ be the minimal $\epsilon_n$-net for $\cP$ w.r.t.\ $\| \cdot \|_\cF$, where $M_n = N_\cF(\epsilon_n,\cP)$. Consider the following scheme:
\begin{itemize}
\item {\em encoder} --- $e_n(Z^n) = \Argmin_{1 \le m \le M_n} \| \emp_{Z^n} - P_m \|_\cF$
\item {\em learner} --- $\wh{f}_n(J,\cdot) = \Argmin_{f \in \cF} P_J(f)$
\end{itemize}
In other words, the encoder finds the element of $\cN_n$ closest to the empirical distribution $P_{Z^n}$ in the $\| \cdot \|_\cF$ norm and transmits its index to the learner. The learner then finds the function in $\cF$ that minimizes the expected loss assuming that the true distribution is the one estimated by the encoder.

It is easy to see that the resulting scheme operates at zero rate. Indeed, from (\ref{eq:entropy_condition}),
$$
\lim_{n \to \infty} \frac{\log M_n}{n}  = \lim_{n \to \infty} \frac{ H_\cF(\epsilon_n,\cP)}{n} = 0.
$$
To bound the expected loss, assume that $P \in \cP$ is the true distribution and let $P_{m^*} \in \cN_n$ be the element of the $\epsilon_n$-net that is closest to $P$, i.e.,
$$
\| P - P_{m^*} \|_\cF = \min_{1 \le m \le M} \| P - P_m \|_\cF \le \epsilon_n.
$$
Let $J = e_n(Z^n)$. We then have
\begin{eqnarray*}
\lefteqn{L(e_n,\wh{f}_n,P) = P(\wh{f}_n)} \\
&&\le \| P - P_J \|_\cF + P_J(\wh{f}_n) \\
&&= \| P - P_J \|_\cF + L^*(\cF,P_J) \\
&&\stackrel{{\rm (a)}}{\le} 2 \| P - P_J \|_\cF + L^*(\cF,P) \\
&&\le 2 \| P - \emp_{Z^n} \|_\cF + 2 \|\emp_{Z^n} - P_J \|_\cF + L^*(\cF,P) \\
&& \stackrel{{\rm (b)}}{\le} 2 \| P - \emp_{Z^n} \|_\cF + 2 \|\emp_{Z^n} - P_{m^*} \|_\cF + L^*(\cF,P) \\
&& \le 4 \| P - \emp_{Z^n} \|_\cF + 2 \| P - P_{m^*} \|_\cF + L^*(\cF,P) \\
&& \le 4 \| P - \emp_{Z^n} \|_\cF + 2\epsilon_n + L^*(\cF,P),
\end{eqnarray*}
where (a) follows from the fact that
$$
\left| L^*(\cF,P) - L^*(\cF,P') \right| \le \| P - P' \|_\cF
$$
for any two $P,P' \in \cP$, and (b) is by construction of the encoder. The remaining steps are consequences of various definitions and the triangle inequality. Taking expectations and the limit as $n \to \infty$, we get
\begin{eqnarray*}
\lefteqn{ \lim_{n \to \infty} \E L(e_n,\wh{f}_n,P)} \\
& \le& 4 \lim_{n \to \infty} \E \|\emp_{Z^n} - P \|_\cF  + 2\lim_{n \to \infty} \epsilon_n + L^*(\cF,P).
\end{eqnarray*}
The first limit on the right-hand side of this inequality is zero by the GC property, while the second one is zero since $\epsilon_n \to 0$. Thus, $\lim_{n \to \infty} \E L(e_n,\wh{f}_n,P) \le L^*(\cF,P)$.
\end{proof}

We can give one particular example when condition (\ref{eq:entropy_condition}) will hold. Given any two probability measures $P,Q$ on $(\cZ,\cA)$, define the {\em variational distance} between them as
$$
\| P - Q \|_V \deq \sup_{\{A_i\} \subseteq \cA}\sum_i |P(A_i) - Q(A_i)|,
$$
where the supremum is over all finite $\cA$-measurable partitions of $\cZ$. Then we can define the covering numbers $N_V(\epsilon,\cP)$ and the Kolmogorov $\epsilon$-entropy $H_V(\epsilon, \cP)$. Now suppose that there exist some constants $C > 0$ and $\alpha > 0$, such that $H_V(\epsilon,\cP) \le C(1/\epsilon)^\alpha$ for small enough $\epsilon$. This will be the case, for instance, when $\cZ$ is a compact subset of a Euclidean space and all $P \in \cP$ have Lipschitz-continuous densities w.r.t.\ some dominating measure $\nu$, and all the Lipschitz constants are all bounded by some $L < + \infty$ \cite{KolTih61}. Then, since $\| P - P' \|_\cF \le B \| P - P' \|_V$ for all $P,P' \in \cP$, we will have $H_V(\epsilon,\cP) \le C'(1/\epsilon)^\alpha$ with $C' = C'(C,B,\alpha)$. Then, choosing $\epsilon_n = 1/\log n$, we will have $H_\cF(\epsilon_n,\cP) \le C' (\log n)^\alpha  = o(n)$.

\subsection{Type II schemes}
\label{ssec:type_II}

The case of Type II schemes is radically different. Whereas in a Type I scheme the encoder can use the training data to estimate the underlying distribution and then communicate its finite-rate description to the learner, in a Type II situation the encoder can only estimate the $Y$-marginal. Unless the distributions in $\cP$ can be reliably identified from their $Y$-marginals (which is a very restrictive condition), the encoder does not have enough ``learning" ability to estimate the underlying distribution. Instead, we will take the following approach.

Given $\Delta \ge 0$, let us suppose that, for each $n$, the encoder can implement a mapping $Y^n \mapsto \wh{Y}^n$, such that, whenever the training data are drawn from some $P \in \cP$ (unknown to both the encoder and the learner), the empirical distribution $\emp_{(X^n,\wh{Y}^n)}$ is, on average, at most $\Delta/4$ away from $P$ in the $\| \cdot \|_\cF$ sense, and that $n^{-1} \log |\wh{Y}^n(\cY^n)| \le R$. Then the encoder communicates a binary description $J$ of $\wh{Y}^n$ at rate $\le R$ to the learning agent, who decodes it to get $\wh{Y}^n$ and then implements the following two-step procedure:
$$
\wh{P} = \argmin_{P \in \cP} \| \emp_{(X^n,\wh{Y}^n)} - P \|_\cF, \qquad \wh{f}_n = \argmin_{f \in \cF} \wh{P}(f).
$$
Then essentially the same technique as in the proof of Theorem~\ref{thm:type_I} will give us $\E L(e_n,\wh{f}_n,P) \le L^*(\cF,P) + \Delta$ for every $P \in \cP$, thus establishing the existence of a scheme operating at rate $R$ and achieving an excess loss of $\le \Delta$ on each $P \in \cP$.

These considerations motivate the definition of the following $n$th-order operational distortion-rate function:
\begin{equation}
\wh{\DD}_n(\cP,\cF,R) \deq \inf_{\wh{Y}^n} \sup_{P \in \cP} \E_P \| \emp_{(X^n,\wh{Y}^n(Y^n))} - P \|_\cF,
\label{eq:n_DRF}
\end{equation}
where the infimum is over all $\wh{Y}^n : \cY^n \to \cY^n$, such that $n^{-1} \log |\{ \wh{Y}^n(y^n) : y^n \in \cY^n \}| \le R$. We also define the limiting operational distortion-rate function
$$
\wh{\DD}(\cP,\cF,R) \deq \lim_{n\to\infty} \wh{\DD}_n(\cP,\cF,R).
$$

We now state the achievability result for Type II schemes in terms of these operational quantities:

\begin{theorem} Given any $R \ge 0$, the pair $(R,4\wh{\DD}(\cP,\cF,R))$ is achievable.
\label{thm:type_II}
\end{theorem}

\begin{proof} For each $n$, let $\wh{Y}^n_* : \cY^n \to \cY^n$ be the encoder that achieves the infimum in (\ref{eq:n_DRF}). Let $\{\wh{y}^n(1),\ldots,\wh{y}^n(M_n)\}$ be some arbitrary enumeration of its codewords. Then we construct the following scheme:
\begin{itemize}
\item {\em encoder} --- $e_n(Y^n) = J$, such that $\wh{Y}^n_*(Y^n) = \wh{y}^n(J)$.
\item {\em learner} --- $\wh{f}_n(J,X^n,\cdot) = \Argmin_{f \in \cF} \wh{P}(f)$, where
$$
\wh{P} = \argmin_{P \in \cP} \| \emp_{(X^n,\wh{y}^n(J))} - P \|_\cF.
$$
\end{itemize}
The scheme $\{(e_n,\wh{f}_n)\}^\infty_{n=1}$ operates at rate $R$ owing to the fact that $n^{-1} \log M_n \le R$. As for the excess loss, we have
\begin{eqnarray*}
\lefteqn{L(e_n,\wh{f}_n,P) = P(\wh{f}_n)} \\
&\le& 2\| P - \wh{P} \|_\cF + L^*(\cF,P) \\
&\le& 2 \| P - \emp_{(X^n,\wh{y}^n(J))} \|_ \cF \\
&& \qquad \qquad + 2 \| \emp_{(X^n,\wh{y}^n(J))} - \wh{P} \|_\cF + L^*(\cF,P) \\
&\le& 4 \| P - \emp_{(X^n,\wh{y}^n(J))} \|_\cF + L^*(\cF,P) \\
&=& 4 \| P - \emp_{(X^n,\wh{Y}^n_*(Y^n))} \|_\cF + L^*(\cF,P).
\end{eqnarray*}
Taking expectations, using the fact that each $\wh{Y}^n_*$ achieves the $n$th-order optimum $\wh{\DD}_n(\cP,\cF,R)$, and then taking the limit as $n \to \infty$, we get
$$
\E L(e_n,\wh{f}_n,P) \le L^*(\cF,P) + 4 \wh{\DD}(\cP,\cF,R), \qquad \forall P \in \cP
$$
which proves the theorem.
\end{proof}

\newcounter{mytempeqncnt}

\begin{figure*}[!t]
\normalsize
\setcounter{mytempeqncnt}{\value{equation}}
\setcounter{equation}{6}
\begin{equation}
\wh{\DD}(\cP,\cF,R) \le \sup_{\alpha > 0} \inf_{\delta > 0} \sup_{P' \in \cM(\cY)} \inf_{Q_{U|Y}: \atop I(P' \times Q_{U|Y}) \le R + \alpha} \sup_{P \in \cP: \atop \| P_Y - P' \|_V \le \delta} \E_{P \times Q_{U|Y}} \| \delta_{(X,U)} - P \|_\cF.
\label{eq:oper_KS_upper}
\end{equation}
\setcounter{equation}{\value{mytempeqncnt}}
\hrulefill
\vspace*{4pt}
\end{figure*}

We would like to express $\wh{\DD}(\cP,\cF,R)$ purely in terms of information-theoretic quantities. It is relatively straightforward to derive an information-theoretic lower bound on $\wh{\DD}(\cP,\cF,R)$. To that end, we will draw upon recent work of Kramer and Savari \cite{KraSav07} on rate-constrained communication of probability distributions. The following properties of $\| \cdot \|_\cF$ are immediate:
\begin{enumerate}
\item $\| P - Q \|_\cF \le 2B$ for all $P,Q \in \cM(\cZ)$.
\item For a fixed $P$, the mapping $Q \mapsto \| Q - P \|_\cF$ is Lipschitz in the variational norm $\| \cdot \|_V$: for all $Q,Q' \in \cM(\cZ)$
$$
\left| \| P - Q \|_\cF - \| P - Q' \|_\cF \right| \le B \| Q - Q' \|_V.
$$
\item The mapping $Q \mapsto \| Q - P \|_\cF$ is convex: for any $Q = \lambda Q_1 + (1-\lambda) Q_2$ with some $\lambda \in [0,1]$ and $Q_1,Q_2 \in \cM(\cZ)$,
$$
\| Q - P \|_\cF \le \lambda \| Q_1 - P \|_\cF + (1-\lambda) \| Q_2 - P \|_\cF.
$$
\end{enumerate}
Then for each $P \in \cP$ the mapping $Q \in \cM(\cZ) \mapsto \| Q - P \|_\cF$ satisfies the requirements listed in Section~III of \cite{KraSav07}. Thus, following Kramer and Savari, we can define, for every $P \in \cP$ and every $R \ge 0$, the distortion-rate function
\begin{equation}
D_\KS(P,\cF,R) \deq \inf  \| P_{XU} - P \|_\cF,
\label{eq:KS}
\end{equation}
where the infimum is over all distributions of the triple $(X,Y,U) \in \cX \times \cY \times \cY$, such that $P_{XY} = P$, $X \to Y \to U$, and $I(Y;U) \le R$. Kramer and Savari deal only with the case when $\cX$ and $\cY$ are both finite. However, it can be shown that (\ref{eq:KS}) is equal to $\wh{\DD}(\cP,\cF,R)$ for general $\cX,\cY$ as well when $\cP$ is a singleton, $\cP = \{P\}$. The proof of this fact (omitted for lack of space) relies on the GC property (\ref{eq:GC}) and on a straightforward extension of the ``piggyback coding" technique of Wyner \cite[Lemma~4.3]{Wyn75} to general (non-finite) alphabets. Moreover, when $|\cP| \ge 2$, we have the following lower bound:

\begin{theorem}  $\wh{\DD}(\cP,\cF,R) \ge \Sup_{P \in \cP} D_\KS(P,\cF,R)$
\end{theorem}

\begin{proof} Fix any code $\wh{Y}^n(\cdot)$ of rate $R$ that achieves $\wh{\DD}_n(\cP,\cF,R)$:
$$
\sup_{P \in \cP} \E_P \| \emp_{(X^n,\wh{Y}^n(Y^n))} - P \|_\cF = \wh{\DD}_n(\cP,\cF,R).
$$
Fix some $P \in \cP$ and let $P_{X_i,Y_i,\wh{Y}_i}$ denote the joint distribution of $(X_i,Y_i,\wh{Y}_i)$ when $(X_1,Y_1),\ldots,(X_n,Y_n)$ are i.i.d.\ according to $P$, and $\wh{Y}_i$ denotes the $i$th component of $\wh{Y}^n(Y^n)$. Also, define the random variables $\bar{X} \in \cX$, $\bar{Y} \in \cY$, and $\bar{U} \in \cY$ with the joint distribution 
$$
P_{\bar{X},\bar{Y},\bar{U}} \deq \frac{1}{n}\sum^n_{i=1} P_{X_i,Y_i,\wh{Y}_i}.
$$
Then $P_{\bar{X},\bar{Y}} = P$ and that $\bar{X} \to \bar{Y} \to \bar{U}$. Using convexity and the fact that $\E \sup_{f \in \cF} [\cdot] \ge \sup_{f \in \cF} E[\cdot]$, we get
$$
\E_P \| \emp_{(X^n,\wh{Y}^n)} - P \|_\cF \ge \| P_{\bar{X},\bar{U}} - P \|_\cF.
$$
That is, $\| P_{\bar{X},\bar{U}} - P \|_\cF \le \wh{\DD}_n(\cP,\cF,R)$ for all $P \in \cP$. Moreover, steps similar to those in \cite[Thm.~1]{KraSav07} give
$$
nR \ge H(\wh{Y}^n) \ge I(Y^n; \wh{Y}^n) \ge \sum^n_{i=1} I(Y_i; \wh{Y}_i) \ge n I(\bar{Y}; \bar{U}).
$$
Thus, we have found a triple of random variables $(\bar{X},\bar{Y},\bar{U}) \in \cX \times \cY \times \cY$, such that: (i) $P_{\bar{X},\bar{Y}} = P$, (ii) $\bar{X} \to \bar{Y} \to \bar{U}$, (iii) $\| P_{\bar{X},\bar{U}} - P \|_\cF \le \wh{\DD}_n(\cP,\cF,R)$, (iv) $I(\bar{Y};\bar{U}) \le R$. Hence, for every $P \in \cP$, $\wh{\DD}_n(\cP,\cF,R) \ge D_\KS(P,\cF,R)$. Taking the supremum over all $P \in \cP$ and then the limit as $n \to \infty$, we get the desired result.
\end{proof}

However, it is not straightforward to derive an information-theoretic upper bound on $\wh{\DD}(\cP,\cF,R)$. This would require constructing a rate-$R$ code that asymptotically achieves $\wh{\DD}(\cP,\cF,R)$. In order to prove achievability, one could take a rate-$R$ code for each ``representative" distribution in $\cP$ (assuming $\cP$ is not too rich, so that it can be represented by a slowly, e.g.,\ subexponentially, growing number of distributions), combine the codes into a union code (which will result in an asymptotically negligible rate overhead), and then devise a rule for mapping the sequence $Y^n$ into one of the codewords. However, the difficulty here is that the encoder can only estimate the $Y$-marginal of the underlying distribution and cannot select the right code based on this information alone. One (suboptimal) strategy is to bound the distortion $\| \emp_{(X^n,\wh{Y}^n)} - P \|_\cF$ by the average of single-letter functions of the form $\rho_{\cF,P}(X_i,\wh{Y}_i) \deq \| \delta_{(X_i,\wh{Y}_i)} - P \|_\cF$, where $\delta_{(X_i,\wh{Y}_i)}$ is the Dirac measure concentrated on $(X_i,\wh{Y}_i)$, and consider the new problem of finding
\begin{equation}
\inf_{\wh{Y}^n} \sup_{P \in \cP} \E_P \left[ \frac{1}{n}\sum^n_{i=1} \rho_{\cF,P}(X_i,\wh{Y}_i) \right]
\label{eq:noisy_sc_bound}
\end{equation}
where the infimum is over all rate-$R$ codes $\wh{Y}^n : \cY^n \to \cY^n$. Then (\ref{eq:noisy_sc_bound}) will be an upper bound on $\wh{\DD}_n(\cP,\cF,R)$.  Note that the problem of minimizing (\ref{eq:noisy_sc_bound}) is an instance of minimax noisy source coding \cite{DemWei03}: given a sequence of i.i.d.\ samples $(X_1,Y_1),(X_2,Y_2),\ldots$ from an unknown $P \in \cP$ and a blocklength $n$, we wish to code $Y^n$ using a rate-$R$ code, such that the sequence $X^n$ is reconstructed from the encoded data with small average $\rho_{P,\cF}(\cdot,\cdot)$ distortion. When $\cY$ is finite, a type-covering argument, as in \cite{DemWei03}, can be used to show (\ref{eq:oper_KS_upper}) at the top of this page (details are omitted for lack of space). Given any $\alpha > 0$, $\delta > 0$, and $P' \in \cM(\cY)$, the second infimum in (\ref{eq:oper_KS_upper}) is over all conditional probability distributions (transition kernels) from $\cY$ to $\cY$, such that the mutual information between $Y$ and $U$ when $Y \sim P'$ and $U|Y \sim Q_{U|Y}$, is at most $R+\alpha$. The inner supremum is over all probability distributions $P \in \cP$, such that their $Y$-marginal $P_Y$ is within $\delta$ from $P'$ in the variational norm $\| \cdot \|_V$, $P \times Q_{U|Y}$ denotes the joint distribution of $X$, $Y$ and $U$ when $(X,Y) \sim P$ and $U|Y \sim Q_{U|Y}$, and $\delta_{(X,U)}$ denotes the Dirac measure concentrated at $(X,U) \in \cX \times \cY$. We leave the problem of tightening (\ref{eq:oper_KS_upper}) for future work. Evidently, the difficulties involved in extending this technique to general $\cY$ are of the same nature as in \cite{DemWei03} and have to do with finding the right topology on $\cM(\cY)$ that would give the same uniform error bounds as for the variational distance in the finite case.

\bibliography{isit09_ale}

\end{document}